\documentclass[12pt]{article}
\usepackage{amsfonts,amssymb,amsmath}
\usepackage[T1]{fontenc}
\usepackage{authblk}

\usepackage{amsthm}
\usepackage{subcaption}
\usepackage{cancel}
\usepackage{xcolor}
\usepackage{graphicx}
\usepackage{tikz}
\usetikzlibrary{arrows,snakes,backgrounds}
\usepackage{hyperref}
\hypersetup{colorlinks=true,citecolor=blue,linkcolor=blue, urlcolor=blue}
\usepackage[font=small, margin=1cm]{caption}

\newtheorem{theorem}{Theorem}
\newtheorem{lemma}{Lemma}

\newcommand{\ZZ}{{\mathbb Z}}

\newcommand{\ra}{\rightarrow}

\newcommand{\g}{\gamma}

\newcommand{\lr}[1]{{\langle {#1}\rangle}}

\def\~#1{\widetilde{#1}}
\def\_#1{_{\mathrm{#1}}}

\usepackage{environ}
\NewEnviron{eqs}{%
\begin{equation}\begin{split}
    \BODY
\end{split}\end{equation}
}

\title{Exact bosonization in arbitrary dimensions}
\author[*]{Yu-An Chen}
\affil[*]{California Institute of Technology, Pasadena, CA 91125, USA}

\begin{document}
\maketitle
\abstract{We extend the previous results of exact bosonization, mapping from fermionic operators to Pauli matrices, in 2d and 3d to arbitrary dimensions. This bosonization map gives a duality between any fermionic system in arbitrary $n$ spatial dimensions and a new class of $(n-1)$-form $\mathbb{Z}_2$ gauge theories in $n$ dimensions with a modified Gauss's law. This map preserves locality and has an explicit dependence on the second Stiefel-Whitney class and a choice of spin structure on the spatial manifold. A new formula for Stiefel-Whitney homology classes on lattices is derived. In the Euclidean path integral, this exact bosonization map is equivalent to introducing a topological ``Steenrod square'' term to the spacetime action.}

\section{Introduction and Summary}

It is well known that every fermionic lattice system in 1d is dual to a lattice system of spins with a $\ZZ_2$ global symmetry (and vice versa). The duality is kinematic (independent of a particular Hamiltonian) and arises from the Jordan-Wigner transformation. Recently it has been shown that any fermionic lattice system in 2d is dual to a $\ZZ_2$ gauge theory with an unusual Gauss's law \cite{CKR18}. The fermion can be identified with the flux excitation of the gauge theory, which is described by the "Chern-Simons-like" term $ i \pi \int A \cup \delta A$ in the spacetime action. The 2d duality is also kinematic. This approach has been generalized to 3d \cite{CK18}. Every fermionic lattice system in 3d is dual to a $\ZZ_2$ 2-form gauge theory with an unusual Gauss's law. Here ``2-form gauge theory'' means that the $\ZZ_2$ variables live on faces ($2$-simplices), while the parameters of the gauge symmetry live on edges ($1$-simplices). 2-form gauge theories in 3+1D have local flux excitations, and the unusual Gauss's law ensures that these excitations are fermions. This Gauss's law can be described by the "Steenrod square" topological action $ i \pi \int B \cup B + B \cup_1 \delta B$. The form of the modified Gauss's law was first observed in \cite{GK16}: a bosonization of fermionic systems in $n$ dimensions must have a global $(n-1)$-form $\ZZ_2$ symmetry with a particular 't Hooft anomaly. The standard Gauss's law leads to a trivial 't Hooft anomaly, so bosonization requires us to modify it in a particular way.

In this paper, we extend these results to arbitrary $n$ dimensions. We show that every fermionic lattice system in $n$-dimension is dual to a $\ZZ_2$ $(n-1)$-form gauge theory with a modified Gauss's law. Our bosonization map is kinematic and local in the same sense as the Jordan-Wigner map\footnote{We only consider the locality preserving map here. Although Jordan-Wigner transformation can map a single fermionic operator into spins, it contains a string operator, which is highly nonlocal. Our bosonization map and Jordan-Wigner transformation both preserve the locality of observables in fermionic systems.}: every local observable on the fermionic side, including the Hamiltonian density, is mapped to a local gauge-invariant observable on the $\ZZ_2$ gauge theory side. In the Euclidean picture, we show explicitly that our bosonization map is equivalent to introducing the topological term in the action:
\begin{equation}
     S_{\text{top}} = i \pi \int_Y (A_{n-1} \cup_{n-3} A_{n-1} + A_{n-1} \cup_{n-2} \delta A_{n-1}),
\end{equation}
where $A_{n-1}$ is a $(n-1)$-form gauge field, a $(n-1)$-cochain $A_{n-1} \in C^{n-1}(Y,\ZZ_2)$, and $Y$ is $(n+1)$-dimensional spacetime manifold. When $A_{n-1}$ is closed, i.e., $\delta A = 0$, this term reduces to the Steenrod square operator \cite{S47}. This ``Steenrod square'' term appears in the Lagrangian of fermionic symmetry-protected-topological (SPT) phases \cite{LZW18} and it is indirectly argued that this term plays the role of statistical transmutation, which makes the theory fermionic \cite{KT17,W17}. Our approach provides an explicit Hamiltonian picture and the bosonization/fermionization procedure is exact, which gives the direct construction for supercohomology fermionic SPT phases. The quantum circuit for the supercohomology SPT ground state and its commuting projector Hamiltonian are derived explicitly in Ref.~\cite{CET20}. All supercohomology fermionic SPT phases in arbitrary dimensions can be constructed from the bosonization map presented in this paper.

There are already several proposals for an analog of the Jordan-Wigner map in arbitrary dimensions \cite{BK02,B05,VC05,B19}. Our construction is most similar to that of Bravyi and Kitaev \cite{BK02}. One advantage of our construction is that we can clearly identify the kind of $n$-dimensional bosonic systems that are dual to fermionic systems: they possess global $(n-1)$-form $\ZZ_2$ symmetry with a specific 't Hooft anomaly, as proposed in \cite{GK16}. It is also manifest in our approach that the bosonization map depends on a choice of spin structure.

\section{Chains, Cochains, and Higher Cup \\ Products}

In this section, we introduce the mathematical tools used in this paper. Our notations and conventions are also described. We will always work with an arbitrary triangulation of a closed simply-connected $n$-dimensional manifold $M_n$ equipped with a branching structure (orientations on edges without forming a loop in any triangle)\footnote{ A direct construction of branching structure is to arbitrarily assign different real numbers on all vertices. For each edge, the arrow is pointed from the smaller number to the larger number.}. 
The vertices, edges, faces, and tetrahedra are denoted $v,e,f,t$, respectively. The general $d$-simplex is denoted as $\Delta_d$. We can label the vertices of $\Delta_d$ as $0,1,2,\dots,d$ such that the directions of edges are from the small number to the larger number. We denote this $d$-simplex as $\Delta_d = \langle 0 1 \dots d \rangle$. Its boundaries are $(d-1)$-simplices $\langle 0, \dots, \hat{i}, \dots, d \rangle$ for $i= 0, 1, \dots, d$, where $\hat{i}$ means $i$ is omitted. A formal sum of $d$-simplices modulo 2 forms an element of the chain $C_d(M_n,\ZZ_2)$.

For every $v$, we define its dual $0$-cochain $\boldsymbol{v}$, which takes value $1$ on $v$, and $0$ otherwise, i.e. $\boldsymbol v(v') = \delta_{v,v'} $. Similarly, $\mathbf{e}$ is an $1$-cochain  $\boldsymbol e(e') = \delta_{e,e'} $, and so forth, i.e., $\boldsymbol \Delta_d$ being a $d$-cochain $\boldsymbol \Delta_d (\Delta'_d) = \delta_{\Delta_d,\Delta'_d}$.  All dual cochains will be denoted in bold. A $d$-cochain $\boldsymbol c_d \in C^d(M_n,\ZZ_2)$ can be identified as a $\ZZ_2$ field living on each $d$-simplex $\Delta_d$, with the value $\boldsymbol c_d (\Delta_d)$.
An evaluation of a cochain $\boldsymbol c$ on a chain $c'$ is the sum of $\boldsymbol c$ evaluated on simplices in $c'$, which is denoted $\int_{c'} \boldsymbol c = \sum_{\Delta \in c'} \boldsymbol c (\Delta)$, . When the integration range is not written, $\boldsymbol c$ is assumed to be the top dimension and $\int \boldsymbol c \equiv \int_{M_n} \boldsymbol c$.

The boundary operator is denoted by $\partial$. For an $n$-simplex $\Delta_n$, $\partial \Delta_n$ consists of all boundary $(n-1)$-simplices of $\Delta_n$:
\begin{eqs}
    \partial (\lr{0,1,2, \dots, d}) = \sum_{i=0}^d \lr{ 0, \dots, \hat{i}, \dots, d }.
\end{eqs}
The coboundary operator is denoted by $\delta$ (not to be confused with the Kronecker delta previously). On a 0-cochain $\boldsymbol v$, $\delta \boldsymbol v$ is an $1$-cochain acting on edges, and is $1$ if $\partial e$ contains $v$ and $0$ otherwise:
\begin{eqs}
    \delta \boldsymbol v (e) = \boldsymbol v (\partial e) =\delta_{v,\partial e}.
\end{eqs}
It is similar for simplices in any dimension. For any $d$-cochain $\boldsymbol c \in C^d(M_n, \ZZ_2)$, its coboundary $\delta \boldsymbol c \in C^{d+1}(M_n, \ZZ_2)$ acting on a $(d+1)$-simplex $\Delta_{d+1} = \lr{0,1,\dots,d+1}$  is defined by:
\begin{eqs}
    \delta \boldsymbol c (\Delta_{d+1}) &\equiv \boldsymbol c (\partial \Delta_{d+1}) \\
    &= \sum_{i=0}^{d+1} \boldsymbol c (\lr{0,\dots,\hat{i},\dots,d+1}).
\end{eqs}

The cup product $\cup$ of a $p$-cochain $\alpha_p$ and a $q$-cochain $\beta_q$ is a $(p+q)$-cochain defined as:
\begin{equation}
    \begin{split}
        [\alpha_p \cup \beta_q ] (\langle 0, 1, \dots, p+q \rangle) &= \alpha_p (\langle 0 1 \dots p \rangle) \beta_q (\langle p, p+1, \dots, p+q \rangle) \\
        &=\alpha_p (\lr{0 \sim p}) \beta_q (\lr{p \sim p+q}),
    \end{split}
\end{equation}
where $ i \sim j$ represents the integers from $i$ to $j$, i.e. $i, i+1, \dots, j$.
The definition of the (higher) cup product $\cup_1$ \cite{GK16, S47} is
\begin{eqs}
    &\left[\alpha_{p} \cup_{1} \beta_{q}\right]( \lr{0, \cdots, p+q-1} )\\
    & =\sum_{i_{0} = 0}^{p+q-1} \alpha_{p}\left( \lr{0 \sim i_{0}, i_{0}+q \sim p+q-1} \right)
    \beta_{q} \left(\lr{ i_{0} \sim i_{0}+q}  \right),
\end{eqs}
and the next cup product $\cup_2$ is
\begin{eqs}
    &\left[\alpha_{p} \cup_{2} \beta_{q}\right]
    ( \lr{0, \cdots, p+q-2} ) \\
    &= \sum_{0 \leq i_{0}<i_{1} \leq p+q-2}
    \alpha_{p}
    \left( \lr{0 \sim i_{0}, i_{1} \sim p+i_{1}-i_{0}-1} \right) \\
    &\quad \quad \quad \quad  \cdot
    \beta_{q}
    \left(\lr{ i_{0} \sim i_{1}, p+i_{1}-i_{0}-1 \sim p+q-2 } \right).
\end{eqs}
The general higher cup product can be expressed as
\begin{equation}
\begin{split}
    &[\alpha_p \cup_a \beta_q] (0,1,\cdots, p+q -a)=\\
    &\sum_{  0 \leq i_0 < i_1 < \cdots < i_a \leq p+q-a } \alpha_p
    (\lr{0\sim i_0,i_1 \sim i_2, i_3 \sim i_4, \cdots}) \times \beta_q(\lr{i_0\sim i_1, i_2 \sim i_3, \cdots } ),
\end{split}
\label{eq: definition of higher cup product}
\end{equation}
where $\{ i_0,i_1, \dots, i_a \}$ are chosen such that the arguments of $\alpha_p$ and $\beta_q$ contain $p+1$ and $q+1$ numbers separately. For example, we have 
\begin{eqs}
    \alpha_2 \cup_1 \beta_1(\lr{012}) = \alpha_2(\lr{012}) \beta_1(\lr{01}) + \alpha_2(\lr{012}) \beta_1(\lr{12}),
\end{eqs}
since the allowed choices are only $(i_0, i_1) = (0,1)$ and $(i_0, i_1) = (1,2)$.
\begin{figure}[htb]
\centering
\resizebox{6cm}{!}{%
\includegraphics[width=0.6\textwidth]{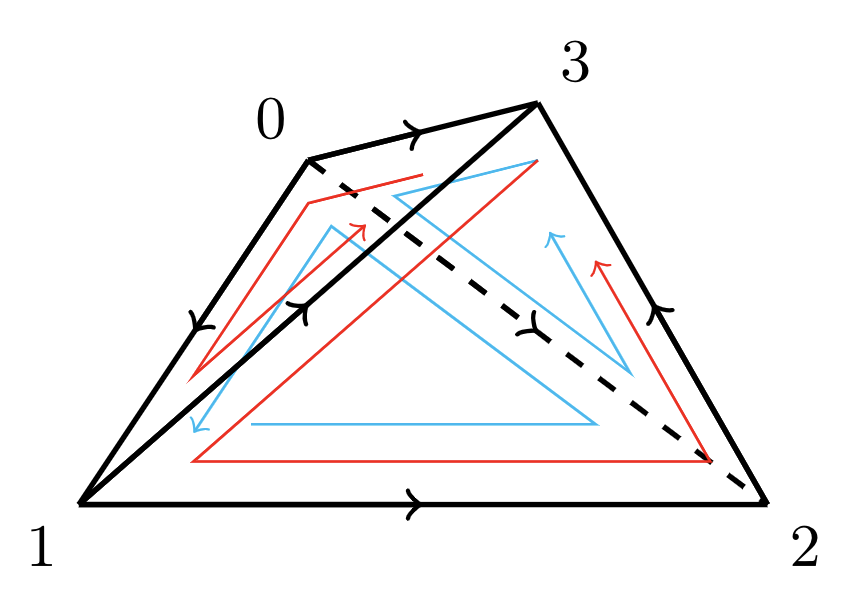}
}
\caption{(Color online) A branching structure on a tetrahedron. The orientation of each face is determined by the right-hand rule. We defined this as the ``$+$'' tetrahedron, the directions of faces $\lr{012}$ and $\lr{023}$ are inward (blue) while the directions of faces $\lr{123}$ and $\lr{013}$ are outward (red). The directions of faces are reversed in the ``$-$'' tetrahedron (mirror image of this tetrahedron) \cite{CK18}.}
\label{fig: tetrahedron}
\end{figure}
Another example is
\begin{eqs}
    \alpha_2 \cup_1 \beta_2 (\lr{0123}) = \alpha_2 (\lr{023}) \beta_2 (\lr{012}) + \alpha_2 (\lr{013}) \beta_2 (\lr{123}),
\end{eqs}
where the choices of $(i_0, i_1)$ are $(0,2)$ and $(1,3)$.
Notice that faces $\lr{023}$ and $\lr{012}$ are inward, while faces $\lr{013}$ and $\lr{123}$ are outward in Fig.~\ref{fig: tetrahedron}. Therefore, the $\cup_1$ product of two $2$-cochains acting on a tetrahedron is the sum of the products of $2$-cochains acting faces with the same orientation (either both inward or both outward). In Section \ref{sec: commutation relations}, this property can be generalized to higher dimensions: the $\cup_{n-2}$ of two $(n-1)$-cochains acting on a $n$-simplex is the sum of the product of $(n-1)$-cochains acting on its boundary $(n-1)$-simplices with the same orientation. This geometry interpretation of higher cup product is crucial since it is further shown that this property coincides with the commutation relations of fermionic hopping operators. The fermionic statistic is captured by higher cup products and this makes it convenient to derive the topological action for fermionic theories. Although not directly used in this paper, a higher cup product of arbitrary cochains has a nice geometrical interpretation \cite{T20}: the higher cup product measures the intersection between dual cells and thickened, shifted version of other dual cells, where the thickening and shifting are determined by the vector frame field. For example, the simplest cup product formula
\begin{eqs}
    \alpha_1 \cup \beta_1 (\lr{012}) = \alpha_1 (\lr{01}) \beta_1(\lr{12}),
\end{eqs}
can be viewed as the intersection point in Fig.~\ref{fig: cup product}.

\begin{figure}[htb]
\centering
\resizebox{6cm}{!}{%
\includegraphics[width=0.6\textwidth]{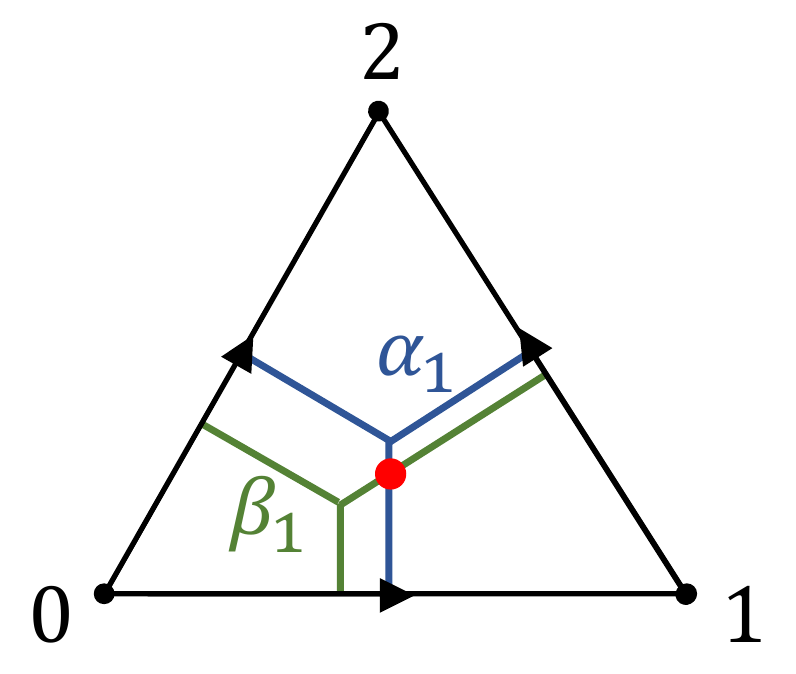}
}
\caption{(Color online) In this figure, $\alpha_1$ is represented by blue segments, dual (attached) to edges $\lr{01}$, $\lr{12}$, and $\lr{02}$ separately. Similarly, $\beta_1$ is represented by green segments with a shifting. The red point is the intersection between the blue segment dual to edge $\lr{01}$ and the green segment dual to edge $\lr{12}$. Therefore, we can read that the cup product of $\alpha_1$ and $\beta_1$ acting on this triangle $\lr{012}$ is $ \alpha_1 (\lr{01}) \beta_1(\lr{12})$.}
\label{fig: cup product}
\end{figure}

It should be emphasized that the cup products satisfy the recursive property:
\begin{equation}
    \alpha \cup_a \beta + \beta \cup_a \alpha = \alpha \cup_{a+1} \delta \beta + \delta \alpha \cup_{a+1} \beta + \delta (\alpha \cup_{a+1} \beta),
\label{eq: recursive property of cup products}
\end{equation}
which can be interpreted as that the non-commutative property of the $\cup_a$ product is equal to the failure of the product rule of the coboundary operation $\delta$ on the $\cup_{a+1}$ product.

Finally, $\Delta^1_n \supset \Delta^2_{n'}$ or $\Delta^2_{n'}\subset \Delta^1_n $ means that the simplex $\Delta^1_n$ contains $\Delta^2_{n'}$ as a subsimplex. A general rule of thumb is that the subset symbol always points to one higher dimension.

\section{Review of Boson-Fermion Duality in (2+1)D and (3+1)D}

We begin by reviewing the duality between fermions and $\ZZ_2$ lattice gauge theory in both two spatial dimensions \cite{CKR18} and three spatial dimensions \cite{CK18}.
On each face $f$ of the 2-manifold $M_2$, we place a single pair of fermionic creation-annihilation operators $c_f,c_f^\dagger$, or equivalently a pair of Majorana fermions:
\begin{eqs}
    \gamma_f = c_f + c^\dagger_f, \quad \gamma'_f = \frac{c_f - c^\dagger_f}{i}
\end{eqs}
The algebra of Majorana fermions is
\begin{equation}
    \{\g_f, \g_{f'}\} = \{\g'_f, \g'_{f'}\} = 2 \delta_{f,f'}, ~ \{\g_f, \g'_{f'}\} = 0
\end{equation}
where $\{A, B\} = AB - BA$ is the anti-commutator. The even fermionic algebra consists of local observables with a trivial fermionic parity (i.e. $P_F= \prod_f (-1)^{c^\dagger_f c_f} = 1$). It is generated by the on-site fermion parity, $$P_f=-i\gamma_f\gamma'_f,$$ and the fermionic hopping operator on every edge $e$, $$S_e=i\gamma_{L(e)}\gamma'_{R(e)},$$ where $L(e)$ and $R(e)$ are faces to the left and right of $e$, with respect to the branching structure of $e$. The commutation relation of hopping operators can be expressed as:
\begin{equation}
    S_e S_{e'} = (-1)^{\int \boldsymbol{e} \cup \boldsymbol{e}' + \boldsymbol{e}' \cup \boldsymbol{e}}S_e S_{e'},
\end{equation}
where the sign from the commutation occurs only when the arrows on the two edges follow head to tail and are on the same triangle, i.e., edges $\{e, e'\}$ being $\{ \lr{01}, \lr{12}\}$ of a triangle $\lr{012}$. In general, for any $1$-cochains $\boldsymbol{\lambda}$ and $\boldsymbol \lambda'$,
\begin{equation}
S_{ \lambda + \lambda'} \equiv (-1)^{\int \boldsymbol \lambda \cup \boldsymbol \lambda'} S_{\lambda'} S_{\lambda}.
\end{equation}
In other words, $S_\lambda$ is the product of $S_e$ over $\{e | \boldsymbol \lambda(e) = 1 \}$ and the sign in front is consistent with the commutation relations.
If we consider the product of fermionic hopping operators on edges around a vertex $v$, the Majorana operators cancel out up to some $P_f$ terms. The two generators $P_f$ and $S_e$ satisfy the following constraint at each vertex $v$ \cite{CKR18}:
\begin{equation}
    (-1)^{\int_{w_2} \boldsymbol v} S_{\delta \boldsymbol v} \prod_{f} P_f^{\int \boldsymbol v \cup \boldsymbol f + \boldsymbol f \cup \boldsymbol v} = 1  
\end{equation}
where $w_2 \in C_0(M_2,\ZZ_2)$ is the $0$-chain which is Poincar\'e dual to the second Stiefel-Whitney cohomology class $\boldsymbol w_2(M_2)$. The explicit expression of $w_2$ is given in Appendix \ref{sec: formula for Stiefel-Whitney homology classes}. We require $M_2$ to be a spin manifold, i.e., the second Stiefel-Whitney class is exact: $w_2 = \partial E$ for some $E \in C_1(M_2,\ZZ_2)$. The $1$-chain $E$ is a choice of the spin structure. The nonexactness of the second Stiefel-Whitney class is the obstruction to determine this $1$-chain $E$, which prevents us from defining a self-consistent bosonization map, which dualizes the even sector of fermionic Hilbert space to a $\ZZ_2$ gauge theory.

The bosonic dual of this system involves $\ZZ_2$-valued spins on the edges of the triangulation. The bosonic algebra are generated by Pauli matrix on edges:
\begin{equation}
X_e=
\begin{bmatrix} 
0 & 1 \\
1 & 0 
\end{bmatrix},
~Y_e=
\begin{bmatrix} 
0 & -i \\
i & 0 
\end{bmatrix},
~Z_e=
\begin{bmatrix} 
1 & 0 \\
0 & -1 
\end{bmatrix}.
\end{equation}
For every face $f$, we define the flux operator:
\begin{equation}
    W_f = \prod_{e \subset f} Z_e,
\end{equation}
and for every edge $e$ we define a unitary operator ${U}_e$ which squares to $1$:
\begin{equation}
    {U}_e =  X_e(\prod_{e^\prime} Z_{e^\prime}^{\int \boldsymbol e' \cup \boldsymbol e}) 
\end{equation}
where $X_e$, $Z_e$ are Pauli matrices acting on a spin at the edge $e$. 
It has been shown in \cite{CKR18} that the sets $\{U_e,~W_f \}$ and $\{S_e,~P_f\}$ satisfy the same commutation relations. The boson-fermion duality map defined on the manifold $M_2$ is:
\begin{equation}
\begin{split}
W_f = \prod_{e \subset f} Z_e&\longleftrightarrow P_f = -i\gamma_f\gamma'_f, \\
{U}_e = X_e (\prod_{e^\prime} Z_{e^\prime}^{\int \boldsymbol e' \cup \boldsymbol e})  &\longleftrightarrow (-1)^{\int_E \boldsymbol e} S_e = (-1)^{\int_E \boldsymbol e} i\gamma_{L(e)}\gamma'_{R(e)},\\
G_v = \prod_{e \supset v} X_e  (\prod_{e^\prime} Z_{e^\prime}^{\int \delta \boldsymbol  v \cup \boldsymbol e'}) &\longleftrightarrow (-1)^{\int_{w_2} \boldsymbol v} S_{ \delta \boldsymbol v} \prod_{f} P_f^{\int \boldsymbol v \cup \boldsymbol f + \boldsymbol f \cup \boldsymbol v} = 1,  \\
\prod_f W_f = 1 &\longleftrightarrow \prod_f P_f
\end{split}
\label{eq: 2d boson-fermion duality}
\end{equation}
where the $0$-chain $w_2 \in C_0(M_2,\ZZ_2)$ is the chain representation of 2nd Stiefel-Whitney class and the $1$-chain $E \in C_1(M_2,\ZZ_2)$ denotes a choice of spin structure ($\partial E = w_2$). For the consistency of this duality map, we need to impose the gauge constraints on bosonic side $\prod_{e \supset v} X_e  (\prod_{e^\prime} Z_{e^\prime}^{\int \delta \boldsymbol  v \cup \boldsymbol e'})= 1$. The gauge invariant subspace in the bosonic Hilbert space is dual to the fermionic system with total fermion parity $\prod_f P_f=1$.

The 3d boson-fermion duality defined on a 3d manifold $M_3$ can be done in a similar way \cite{CK18}. The only difference is that the fermions $\g_t,~\g_t^\prime$ are at the center of tetrahedra $t$ and Pauli operators $X_f,Z_f$ live on faces $f$.  In three spaitial dimensions, any fermionic system can be mapped to a $2$-form $\ZZ_2$ gauge theory on the 3d lattice. The duality dictionary becomes:

\begin{equation}
\begin{split}
W_t = \prod_{f \subset t} Z_f&\longleftrightarrow P_t = -i\gamma_t\gamma'_t, \\
{U}_f = X_f (\prod_{f^\prime} Z_{f^\prime}^{\int \boldsymbol f^\prime \cup_1 \boldsymbol f})
&\longleftrightarrow (-1)^{\int_E \boldsymbol f} S_f = (-1)^{\int_E \boldsymbol f} i\gamma_{L(f)}\gamma'_{R(f)},\\
G_e  = \prod_{f \supset e} X_f  (\prod_{f^\prime} Z_{f^\prime}^{\int \delta \boldsymbol e \cup_1 \boldsymbol f'}) &\longleftrightarrow (-1)^{\int_{w_2} \boldsymbol e} S_{\delta \boldsymbol e} \prod_{t} P_t^{\int \boldsymbol e \cup_1 \boldsymbol t + \boldsymbol t \cup_1 \boldsymbol e} = 1, \\
\prod_t W_t = 1 &\longleftrightarrow \prod_t P_t
\end{split}
\label{eq: 3d boson-fermion duality}
\end{equation}
where the $1$-chain $w_2 \in C_1(M_3,\ZZ_2)$ is the chain representative of the second Stiefel-Whitney class, and the $2$-chain $E \in C_2(M_3,\ZZ_2)$ is a choice of spin structure ($\partial E = w_2$).

\section{Exact bosonization in $n$ dimensions}

From the 2d and 3d formulae \eqref{eq: 2d boson-fermion duality} and \eqref{eq: 3d boson-fermion duality}, it is very natural to conjecture the $n$-dimensional boson-fermion duality. Consider a spin manifold $M_n$ in spatial $n$ dimensions. The fermions live at the center $n$-simplices, i.e. $\g_{\Delta_n},\g'_{\Delta_n}$ for each $\Delta_n$.  The Pauli matrices live on $(n-1)$-simplices, i.e. $X_{\Delta_{n-1}}$ and $Z_{\Delta_{n-1}}$ for each $\Delta_{n-1}$. The $n$-dimensional boson-fermion duality should be:

\hspace*{-2cm}\vbox{
\begin{equation}
\begin{split}
& W_{\Delta_n} \equiv \prod_{\Delta_{n-1} \subset \Delta_n} Z_{\Delta_{n-1}}
\longleftrightarrow P_t = -i\gamma_{\Delta_n}\gamma'_{\Delta_n}, \\[10pt]
&{U}_{\Delta_{n-1}} \equiv X_{\Delta_{n-1}}  \left(\prod_{{\Delta_{n-1}}^\prime}  Z_{{\Delta_{n-1}}^\prime}^{\int \boldsymbol {\Delta_{n-1}}^\prime \cup_{n-2} \boldsymbol {\Delta_{n-1}}}\right)  \\
& \longleftrightarrow (-1)^{\int_E \boldsymbol {\Delta_{n-1}}} S_{\Delta_{n-1}} =     (-1)^{\int_E \boldsymbol {\Delta_{n-1}}} i\gamma_{L({\Delta_{n-1}})}\gamma'_{R({\Delta_{n-1}})},\\[10pt]
&G_{\Delta_{n-2}} \equiv \prod_{{\Delta_{n-1}} \supset {\Delta_{n-2}}} X_{\Delta_{n-1}} \left(\prod_{{\Delta_{n-1}}^\prime} Z_{{\Delta_{n-1}}^\prime}^{\int \delta \boldsymbol {\Delta_{n-2}} \cup_{n-2} \boldsymbol {\Delta_{n-1}}^\prime}\right) \\
&\longleftrightarrow (-1)^{\int_{w_2} \boldsymbol {\Delta_{n-2}}}   S_{\delta \boldsymbol \Delta_{n-2}} \prod_{{\Delta_{n}}} P_{\Delta_{n}}^{\int \boldsymbol {\Delta_{n-2}} \cup_{n-2} \boldsymbol {\Delta_{n}} + \boldsymbol {\Delta_{n}} \cup_{n-2} \boldsymbol {\Delta_{n-2}}} = 1, \\[10pt]
&\prod_{\Delta_{n}} W_{\Delta_{n}} = 1 \longleftrightarrow \prod_{\Delta_{n}} P_{\Delta_{n}}
\end{split}
\label{eq: nd boson-fermion duality}
\end{equation}}

\noindent where $w_2 \in C_{n-2}(M_n,\ZZ_2)$ is the chain representative of the second Stiefel-Whitney class, $E \in C_{n-1}(M_n,\ZZ_2)$ denotes a choice of spin structure ($\partial E = w_2$), and for general $(n-1)$-cochain $\boldsymbol \lambda_{n-1}$ and $ \boldsymbol \lambda'_{n-1}$, the product of $S$ operators is defined as
\begin{equation}
    S_{ \lambda_{n-1} + \lambda'_{n-1}} \equiv (-1)^{\int \boldsymbol \lambda_{n-1} \cup_{n-2} \boldsymbol \lambda'_{n-1}} S_{\lambda'_{n-1}} S_{\lambda_{n-1}}.
\label{eq: S lambda}
\end{equation}
This $n$-dimensional boson-fermion duality \eqref{eq: nd boson-fermion duality} is the main theorem of this paper, which will be proved by the end of this section.

\subsection{Commutation relations} \label{sec: commutation relations}
 
 Consider an $n$-simplex $\Delta_n = \langle 012 \dots n \rangle$. Its boundary contains all $(n-1)$-simplex $ (\partial \Delta_n)^i = \langle 0 \dots \hat{i} \dots n \rangle$ where $\hat{i}$ means the vertex $i$ is omitted. We define the orientation of $(\partial \Delta_n)^i$ as $O((\partial \Delta_n)^i) = (-1)^i$. For ``$+$''-oriented $\Delta_n$, if $O((\partial \Delta_n)^i)=1$, the boundary $(\partial \Delta_n)^i$ is outward, and if $O((\partial \Delta_n)^i)=-1$, the boundary $(\partial \Delta_n)^i$ is inward. For ``$-$''-oriented $\Delta_n$, the inward and outward boundaries are opposite. $S_{\Delta_{n-1}}$ and $S_{\Delta'_{n-1}}$ anti-commute only when $\Delta_{n-1}$ and $\Delta'_{n-1}$ are both inward or both outward boundaries of some $n$-simplex, i.e. $\Delta_{n-1}, \Delta'_{n-1} \in \partial \Delta_n$. We are going to prove that this is equivalent to
\begin{equation}
    S_{\Delta_{n-1}} S_{\Delta'_{n-1}} = (-1)^{\int \boldsymbol {\Delta_{n-1}} \cup_{n-2} \boldsymbol {\Delta'_{n-1}}  + \boldsymbol {\Delta'_{n-1}} \cup_{n-2} \boldsymbol {\Delta_{n-1}} } S_{\Delta'_{n-1}} S_{\Delta_{n-1}}.
\label{eq: commutation relation of fermionic hopping}
\end{equation}
From the definition of the higher cup product \eqref{eq: definition of higher cup product}, we have
\begin{equation}
    \begin{split}
        &[\boldsymbol {\Delta_{n-1}} \cup_{n-2} \boldsymbol {\Delta'_{n-1}}](0,1,\cdots,n) \\
        =&\sum_{0 \leq i_0 < i_1 < \cdots < i_{n-2} \leq n} \boldsymbol {\Delta_{n-1}}(0\sim i_0,i_1 \sim i_2, i_3 \sim i_4, \cdots) \boldsymbol {\Delta'_{n-1}}(i_0\sim i_1, i_2 \sim i_3, \cdots) \\
        =& \sum_{0 \leq j_1 < j_2 \leq n| j_1,j_2 \in \text{even} } \boldsymbol {\Delta_{n-1}} (\langle 0 \dots \hat{j_2} \dots n \rangle) \boldsymbol {\Delta'_{n-1}} (\langle 0 \dots \hat{j_1} \dots n \rangle) \\
        &+ \sum_{0 \leq k_1 < k_2 \leq n| k_1,k_2 \in \text{odd} } \boldsymbol {\Delta_{n-1}} (\langle 0 \dots \hat{k_1} \dots n \rangle) \boldsymbol {\Delta'_{n-1}} (\langle 0 \dots \hat{k_2} \dots n \rangle).
    \end{split}
\label{eq: definition of cup n-2}
\end{equation}
The $\cup_{n-2}$ only contains the product of boundaries $\Delta^i_{n-1}$ with the same orientation (inward or outward) and each pair of $\Delta^i_{n-1}, \Delta^{i'}_{n-1}$ with the same orientation appears exactly once. Therefore, the $\cup_{n-2}$ expression in \eqref{eq: commutation relation of fermionic hopping} captures the commutation relations of fermionic hopping operators $S_{\Delta_{n-1}}$. It is easy to check that bosonic operators $U_{\Delta_{n-1}}$ satisfy the same commutation relations:
\begin{equation}
    U_{\Delta_{n-1}} U_{\Delta'_{n-1}} = (-1)^{\int \boldsymbol {\Delta_{n-1}} \cup_{n-2} \boldsymbol {\Delta'_{n-1}}  + \boldsymbol {\Delta'_{n-1}} \cup_{n-2} \boldsymbol {\Delta_{n-1}} } U_{\Delta'_{n-1}} U_{\Delta_{n-1}}.
\label{eq: commutation relation of bosonic hopping}
\end{equation}
Therefore, $\{S_{\Delta_{n-1}}, P_{\Delta_n} \}$ and $\{U_{\Delta_{n-1}}, W_{\Delta_n} \}$ in \eqref{eq: nd boson-fermion duality} have the same commutation relations.

\subsection{Gauge constraints}

In this section, we will derive the constraints on fermionic operators:
\begin{equation}
     (-1)^{\int_{w_2} \boldsymbol \Delta_{n-2}} S_{\delta {\Delta_{n-2}}} \prod_{{\Delta_{n}}} P_{\Delta_{n}}^{\int \boldsymbol \Delta_{n-2} \cup_{n-2} \boldsymbol \Delta_{n} + \boldsymbol \Delta_{n} \cup_{n-2} \boldsymbol \Delta_{n-2}} = 1.
\end{equation}
This follows directly from the following two lemmas.

\begin{lemma}
    The Majorana operators in $S_{\delta {\Delta_{n-2}}}$ cancel out with Majorana operators in $\prod_{{\Delta_{n}}} P_{\Delta_{n}}^{\int \boldsymbol \Delta_{n-2} \cup_{n-2} \boldsymbol \Delta_{n} + \boldsymbol \Delta_{n} \cup_{n-2} \boldsymbol \Delta_{n-2}}$.
\label{lemma: S P cancel}
\end{lemma}

\begin{lemma}
    The sign difference of $S_{\delta {\Delta_{n-2}}}$ and the product of on-site fermion parities $\prod_{{\Delta_{n}}} P_{\Delta_{n}}^{\int \boldsymbol \Delta_{n-2} \cup_{n-2} \boldsymbol \Delta_{n} + \boldsymbol \Delta_{n} \cup_{n-2} \boldsymbol \Delta_{n-2}}$ is $-(-1)^{\sum^d_{i=1} \int \boldsymbol \Delta^{i-1}_{n-1} \cup_{n-2} \boldsymbol \Delta^{i}_{n-1} }$ where we order $(n-1)$-simplices $\{ \Delta_{n-1} | \Delta_{n-1} \supset {\Delta_{n-2}} \}$ counterclockwise as  $\boldsymbol \Delta^1_{n-1}, \boldsymbol \Delta^2_{n-1}, \dots, \boldsymbol\Delta^{d-1}_{n-1}, \boldsymbol\Delta^d_{n-1} \equiv \boldsymbol\Delta^0_{n-1}$, as shown in Fig.~\ref{fig: S_prod2}. This sign is a chain representative of the second Stiefel-Whitney class:
    \begin{equation}
        -(-1)^{\sum^d_{i=1} \int \boldsymbol \Delta^{i-1}_{n-1} \cup_{n-2} \boldsymbol \Delta^{i}_{n-1} } = (-1)^{\int_{w_2} \Delta_{n-2} }.
    \label{eq: S P sign}
    \end{equation}
\label{lemma: S P sign}
\end{lemma}

\begin{proof}[Proof of Lemma \ref{lemma: S P cancel}]

Let us denote $\Delta_n = \langle 01 \dots n \rangle$ formed by $\Delta_{n-2}$ and two $(n-1)$-simplex $\Delta^L_{n-1}$ and $\Delta^R_{n-1}$, shown in Fig.~\ref{fig: S_prod1}(a). We know that $S_{\delta {\Delta_{n-2}}}$ contains $\g_{\Delta_n} \g'_{\Delta_n}$ if and only if $\Delta^L_{n-1}, \Delta^R_{n-1}$ are one inward boundary and one outward boundary of $n$-simplex $\Delta_n$, as indicated in Fig. \ref{fig: S_prod1}(b) and (c).

For the product of $P_{\Delta_n}$, we simplify the integral as
\begin{equation}
    \begin{split}
        \int \boldsymbol \Delta_{n-2} \cup_{n-2} \boldsymbol \Delta_{n} + \boldsymbol \Delta_{n} \cup_{n-2} \boldsymbol \Delta_{n-2} = \int \delta  \boldsymbol \Delta_{n-2} \cup_{n-1} \boldsymbol \Delta_{n},
    \end{split}
\label{eq: P integrand}
\end{equation}
where we have used the property $\delta (\alpha \cup_{n-1} \beta) = \delta \alpha \cup_{n-1} \beta + \alpha \cup_{n-1} \delta \beta + \alpha \cup_{n-2} \beta + \beta \cup_{n-2} \alpha$ and $\delta \Delta_n = 0$ (since $n$ is the top dimension).
The integral \eqref{eq: P integrand} has only the contribution from $\Delta_n = \langle 01 \dots n \rangle$:
\begin{equation}
    \begin{split}
        &\int \boldsymbol \Delta_{n-2} \cup_{n-2} \boldsymbol \Delta_{n} + \boldsymbol \Delta_{n} \cup_{n-2} \boldsymbol \Delta_{n-2}\\
        =&[(\boldsymbol \Delta^L_{n-1}+ \boldsymbol \Delta^R_{n-1}) \cup_{n-1} \boldsymbol {\Delta_{n}}] (\langle 01 \dots n \rangle) \\
        =&\sum_{0 \leq i_0 < i_1 < \cdots < i_{n-1} \leq n} (\boldsymbol \Delta^L_{n-1}+ \boldsymbol \Delta^R_{n-1}) (0\sim i_0,i_1 \sim i_2, i_3 \sim i_4, \cdots) \boldsymbol \Delta_{n}(i_0\sim i_1, i_2 \sim i_3, \cdots) \\
        =&\sum_{0\leq j \leq n | j \in \text{odd}} (\boldsymbol \Delta^L_{n-1}+ \boldsymbol \Delta^R_{n-1}) (\langle 0 \dots \hat{j} \dots n \rangle)\boldsymbol \Delta_{n} (\langle 01 \dots n \rangle) \\
        =&\sum_{0\leq j \leq n | j \in \text{odd}} (\boldsymbol \Delta^L_{n-1}+ \boldsymbol \Delta^R_{n-1}) (\langle 0 \dots \hat{j} \dots n \rangle)
    \end{split}
\end{equation}
which is $1$ if and only $\Delta^L_{n-1}, \Delta^R_{n-1}$ are one inward boundary and one outward boundary of the $n$-simplex $\Delta_n$. This shows that product of $P_{\Delta_n}$ contain $P_{\Delta_n} \sim \g_{\Delta_n} \g'_{\Delta_n}$ if and only if $\Delta^L_{n-1}, \Delta^R_{n-1}$ are one inward boundary and one outward boundary of the $n$-simplex $\Delta_n$. This cancels out with $S_{\delta {\Delta_{n-2}}}$ exactly.
\end{proof}

\begin{figure}[htb]
\centering
\includegraphics[width=1\textwidth]{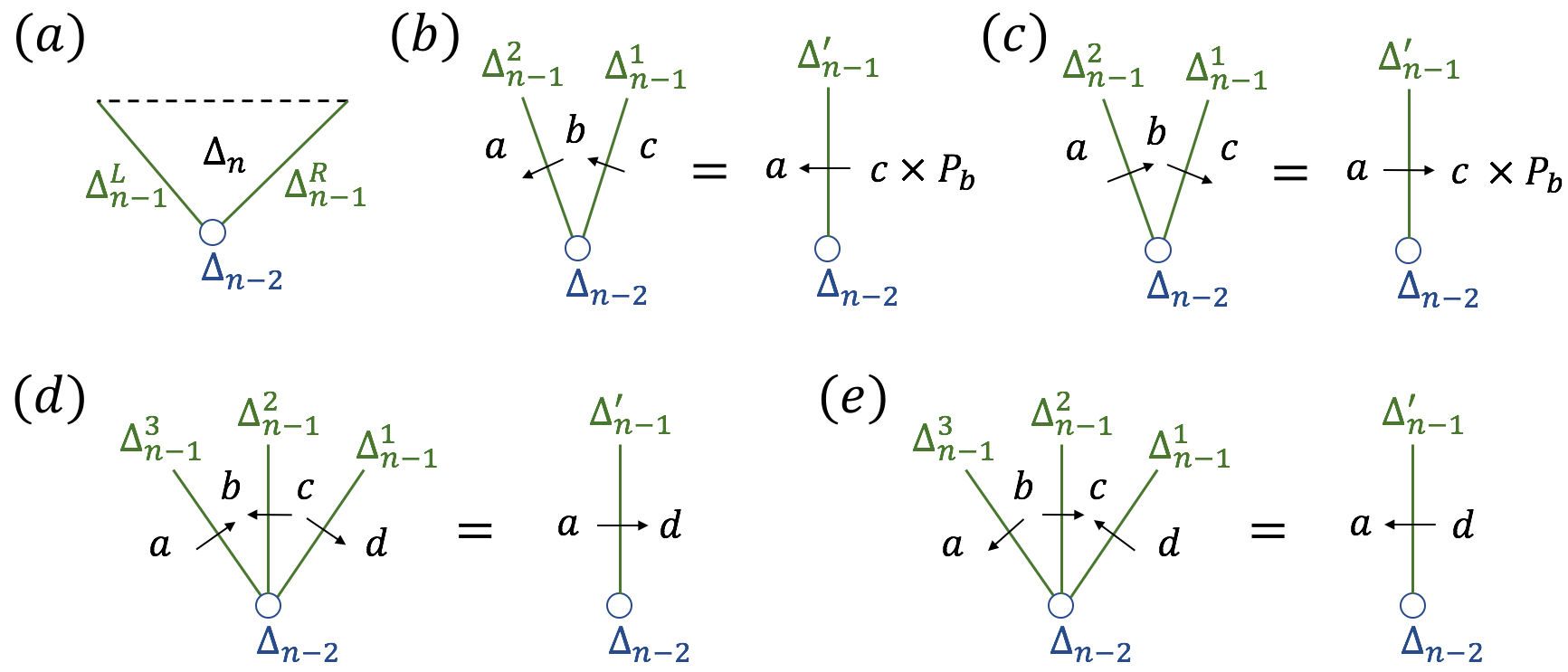}
\caption{(Color online) (a) The $n$-simplex $\Delta_n$ is formed by the $(n-1)$-simplex $\Delta_{n-2}$ and two $(n-1)$-simplex $\Delta^L_{n-1}$ and $\Delta^R_{n-1}$. (b) The product of $S_{\Delta_{n-2}}$ is $(i \g_b \g'_a) (i \g_c \g'_b) = (i \g_c \g'_a) (-i\g_b \g'_b)=(i \g_c \g'_a) P_b $. (c) The product of $S_{\Delta_{n-2}}$ is $(i \g_a \g'_b) (i \g_b \g'_c) = (i \g_a \g'_c) (-i\g_b \g'_b)=(i \g_a \g'_c) P_b $. (d) The product of $S_{\Delta_{n-2}}$ is $(i \g_a \g'_b) (i \g_c \g'_b) (i \g_c \g'_d) = i \g_a \g'_d$. (e) The product of $S_{\Delta_{n-2}}$ is $(i \g_b \g'_a) (i \g_b \g'_c) (i \g_d \g'_c) = i \g_d \g'_a$.} 
\label{fig: S_prod1}
\end{figure}

\begin{proof}[Proof of Lemma \ref{lemma: S P sign}]
\begin{figure}[htb]
\centering
\includegraphics[width=0.8\textwidth]{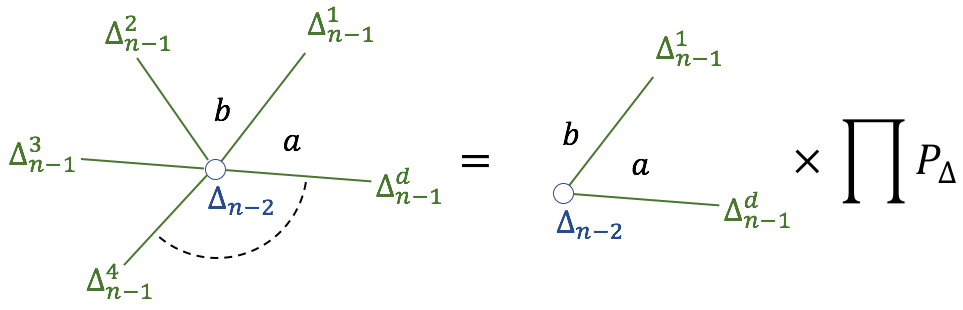}
\caption{(Color online) By the operations defined in Fig.~\ref{fig: S_prod1}, we can simplify the product $S_{\Delta^d_{n-1}} \cdots S_{\Delta^2_{n-1}} S_{\Delta^1_{n-1}} = S_{\Delta^d_{n-1}} S_{\Delta^1_{n-1}} \prod_{{\Delta_{n}}\neq a,b} P_{\Delta_{n}}^{\int \boldsymbol \Delta_{n-2} \cup_{n-2} \boldsymbol \Delta_{n} + \boldsymbol \Delta_{n} \cup_{n-2} \boldsymbol \Delta_{n-2}}$. }
\label{fig: S_prod2}
\end{figure}
We compare the signs between
\begin{equation}
    S_{\delta {\Delta_{n-2}}} = (-1)^{\sum_{\boldsymbol \Delta_{n-1} < \boldsymbol \Delta'_{n-1} | \Delta_{n-1}, \Delta'_{n-1} \supset \Delta_{n-2} } \boldsymbol \Delta_{n-1} \cup_{n-2} \boldsymbol \Delta'_{n-1} } \prod_{\Delta_{n-1} \supset \Delta_{n-2}} S_{\Delta_{n-1}}
\label{eq: prod of S}
\end{equation}
and
\begin{equation}
    \prod_{{\Delta_{n}}} P_{\Delta_{n}}^{\int \boldsymbol \Delta_{n-2} \cup_{n-2} \boldsymbol \Delta_{n} + \boldsymbol \Delta_{n} \cup_{n-2} \boldsymbol \Delta_{n-2}}
\end{equation}
where we have used the definition of $S_{ \lambda_{n-1}}$ in \eqref{eq: S lambda}.
As shown in Fig. \ref{fig: S_prod2},
\begin{equation}
    \begin{split}
        S_{\Delta^d_{n-1}} \cdots S_{\Delta^2_{n-1}} S_{\Delta^1_{n-1}} = S_{\Delta^d_{n-1}} S_{\Delta^1_{n-1}} \prod_{{\Delta_{n}}\neq a,b} P_{\Delta_{n}}^{\int \boldsymbol \Delta_{n-2} \cup_{n-2} \boldsymbol \Delta_{n} + \boldsymbol \Delta_{n} \cup_{n-2} \boldsymbol \Delta_{n-2}}.
    \end{split}
\end{equation}
We can check that 
\begin{equation}
    S_{\Delta^d_{n-1}} S_{\Delta^1_{n-1}}= -(-1)^{\int \boldsymbol \Delta^1_{n-1}\cup_{n-2} \boldsymbol \Delta^d_{n-1}   + \boldsymbol \Delta^d_{n-1}\cup_{n-2} \boldsymbol \Delta^1_{n-1}}  \prod_{{\Delta_{n}}= a,b} P_{\Delta_{n}}^{\int \boldsymbol {\Delta_{n-2}} \cup_{n-2} \boldsymbol \Delta_{n} + \boldsymbol \Delta_{n} \cup_{n-2} \boldsymbol \Delta_{n-2}},
\end{equation}
since $S_{\Delta^d_{n-1}} S_{\Delta^1_{n-1}}$ is $1$ (or $-P_a P_b$) if $\Delta^1_{n-1}, \Delta^d_{n-1}$ are both inward or outward (or one inward and one outward) in $\Delta_n = a$. Therefore,
\begin{eqs}
    S_{\Delta^d_{n-1}} \cdots S_{\Delta^2_{n-1}} S_{\Delta^1_{n-1}} =&  -(-1)^{\int \boldsymbol \Delta^1_{n-1}\cup_{n-2} \boldsymbol \Delta^d_{n-1}   + \boldsymbol \Delta^d_{n-1}\cup_{n-2} \boldsymbol \Delta^1_{n-1} } \\
    & \quad \cdot \prod_{{\Delta_{n}}} P_{\Delta_{n}}^{\int \boldsymbol \Delta_{n-2} \cup_{n-2} \boldsymbol {\Delta_{n}} + \boldsymbol \Delta_{n} \cup_{n-2} \boldsymbol \Delta_{n-2}}.
\end{eqs}
Together with \eqref{eq: prod of S}, we have
\begin{equation}
    \begin{split}
    &S_{\delta {\Delta_{n-2}}} \prod_{{\Delta_{n}}} P_{\Delta_{n}}^{\int \boldsymbol \Delta_{n-2} \cup_{n-2} \boldsymbol \Delta_{n} + \boldsymbol \Delta_{n} \cup_{n-2} \boldsymbol \Delta_{n-2}}\\
    =& (-1)^{\int \boldsymbol \Delta^{1}_{n-1} \cup_{n-2} \boldsymbol \Delta^{d}_{n-1} + \sum^d_{i=2} \int \boldsymbol \Delta^{i-1}_{n-1} \cup_{n-2} \boldsymbol \Delta^{i}_{n-1} } \left(-(-1)^{\int \boldsymbol \Delta^1_{n-1}\cup_{n-2} \boldsymbol \Delta^d_{n-1}   + \boldsymbol \Delta^d_{n-1}\cup_{n-2} \boldsymbol \Delta^1_{n-1}} \right) \\
    =& -(-1)^{\sum^d_{i=1} \int \boldsymbol \Delta^{i-1}_{n-1} \cup_{n-2} \boldsymbol \Delta^{i}_{n-1} }.
    \end{split}
\end{equation}
From the definition of $\cup_{n-2}$ product \eqref{eq: definition of cup n-2},
\begin{equation}
    \begin{split}
        &\sum^d_{i=1} \int \boldsymbol \Delta^{i-1}_{n-1} \cup_{n-2} \boldsymbol \Delta^{i}_{n-1} \\
        =& \sum^d_{i=1} \sum_{\Delta_n} \boldsymbol \Delta^{i-1}_{n-1} \cup_{n-2} \boldsymbol \Delta^{i}_{n-1} (\Delta_n) \\
        =& \sum_{``-"\text{-oriented } \Delta_n = \langle 0 \dots n \rangle } \sum_{j_1 < j_2 | j_1, j_2 \in \text{even}} \boldsymbol \Delta_{n-2} ( \langle 0 \cdots \hat{j_1} \cdots \hat{j_2} \cdots n \rangle) \\
        &+ \sum_{``+"\text{-oriented } \Delta_n = \langle 0 \dots n \rangle } \sum_{k_1 < k_2 | k_1, k_2 \in \text{odd}} \boldsymbol \Delta_{n-2} ( \langle 0 \cdots \hat{k_1} \cdots \hat{k_2} \cdots n \rangle),
    \end{split}
    \label{eq: prod of cup}
\end{equation}
where the relation between $\Delta_{n-2}$, $\Delta^i_{n-1}$, and $\Delta_n$ is shown in Fig.~\ref{fig: S_prod2}. The distinct orientations of $``-"\text{-oriented } \Delta_n$ and $``+"\text{-oriented } \Delta_n$ in the summation come from the fact that $j_1,j_2$ and $k_1,k_2$ in \eqref{eq: definition of cup n-2} have opposite orders.
Eq. \eqref{eq: prod of cup} is related to $w_2$ by the following lemma \ref{lemma: w2}, which is proved in appendix \ref{sec: formula for Stiefel-Whitney homology classes}. Therefore, we derive
\begin{equation}
    -(-1)^{\sum^d_{i=1} \int \boldsymbol \Delta^{i-1}_{n-1} \cup_{n-2} \boldsymbol \Delta^{i}_{n-1} } = (-1)^{\int_{w_2} \boldsymbol \Delta_{n-2} }.
\end{equation}

\begin{lemma}
    In $n$-dimension manifold with triangulation and branching structure, the homology class of $w_2$ can be represented by a $(n-2)$-chain $w_2 \in C_{n-2}(M_n, \ZZ_2)$:
    \begin{equation}
        w_2 = \sum_{\Delta_{n-2}} c(\Delta_{n-2}) \Delta_{n-2}
    \end{equation}
    where
    \begin{equation}
        \begin{split}
            c(\Delta_{n-2}) =& 1 + \sum_{``-"\text{-oriented } \Delta_n = \langle 0 \dots n \rangle } \sum_{j_1 < j_2 | j_1, j_2 \in \text{even}} \boldsymbol \Delta_{n-2} ( \langle 0 \cdots \hat{j_1} \cdots \hat{j_2} \cdots n \rangle) \\
            &+ \sum_{``+"\text{-oriented } \Delta_n = \langle 0 \dots n \rangle } \sum_{k_1 < k_2 | k_1, k_2 \in \text{odd}} \boldsymbol \Delta_{n-2} ( \langle 0 \cdots \hat{k_1} \cdots \hat{k_2} \cdots n \rangle).
        \end{split}
    \label{eq: c Delta}
    \end{equation}
\label{lemma: w2}
\end{lemma}
\end{proof}

We can modify the sign of $S_{\Delta_{n-1}}$ as
\begin{equation}
     S^E_{\Delta_{n-1}} \equiv (-1)^{\int_E \boldsymbol \Delta_{n-1}}S_{\Delta_{n-1}}
\end{equation}
where $E \in C_{n-1}(M_n,\ZZ_2)$ is a choice of spin structure satisfying $\partial E = w_2$. In these modified operators, the constraint on the fermionic operator becomes
\begin{equation}
    S^E_{\delta {\Delta_{n-2}}} \prod_{{\Delta_{n}}} P_{\Delta_{n}}^{\int \boldsymbol \Delta_{n-2} \cup_{n-2} \boldsymbol \Delta_{n} + \boldsymbol \Delta_{n} \cup_{n-2} \boldsymbol \Delta_{n-2}} = 1,
\end{equation}
which is mapped to 
\begin{equation}
    \begin{split}
        G_{\Delta_{n-2}} =& U_{\delta {\Delta_{n-2}}} \prod_{{\Delta_{n}}} W_{\Delta_{n}}^{\int \boldsymbol \Delta_{n-2} \cup_{n-2} \boldsymbol \Delta_{n} + \boldsymbol \Delta_{n} \cup_{n-2} \boldsymbol \Delta_{n-2}} \\
        =& \prod_{{\Delta_{n-1}} \supset {\Delta_{n-2}}} X_{\Delta_{n-1}}  (\prod_{{\Delta_{n-1}}^\prime} Z_{{\Delta_{n-1}}^\prime}^{\int \delta \boldsymbol {\Delta_{n-2}} \cup_{n-2} \boldsymbol {\Delta_{n-1}}^\prime}),
    \end{split}
\end{equation}
where $U_{\lambda_{n-1}}$ is defined by
\begin{equation}
    U_{ \lambda_{n-1} + \lambda'_{n-1}} \equiv (-1)^{\int \boldsymbol \lambda_{n-1} \cup_{n-2} \boldsymbol \lambda'_{n-1}} U_{\lambda'_{n-1}} U_{\lambda_{n-1}},
\label{eq: U lambda}
\end{equation}
and it can be calculated directly from $U_{\Delta_{n-1}}$ defined in \eqref{eq: nd boson-fermion duality}:
\begin{eqs}
    {U}_{\lambda_{n-1}} \equiv \prod_{\Delta_{n-1}| \boldsymbol \lambda(\Delta_{n-1})=1} X_{\Delta_{n-1}}  \left(\prod_{{\Delta_{n-1}}^\prime}  Z_{{\Delta_{n-1}}^\prime}^{\int \boldsymbol {\Delta_{n-1}}^\prime \cup_{n-2} \boldsymbol {\lambda_{n-1}}}\right).
\end{eqs}
We need to impose this gauge constraint $G_{\Delta_{n-2}} = 1$ on bosonic operators for every $(n-2)$-simplex $\Delta_{n-2}$.

We also need to impose the even total parity constraint for fermions
\begin{equation}
    \prod_{\Delta_{n}} P_{\Delta_{n}} = 1
\end{equation}
since it is mapped to the bosonic operator $\prod_{\Delta_{n}} W_{\Delta_{n}} = 1$. After imposing the gauge constraints, the $n$-dimensional boson-fermion duality \eqref{eq: nd boson-fermion duality} is completed.

\section{Modified Gauss's law and Euclidean action}

\subsection{Gauss's law as boundary anomaly}

First, we consider the standard $(n-1)$-form $\ZZ_2$ lattice gauge theory on the $n$-dimensional manifold $M_n$:
\begin{equation}
    H^0 = - J_1 \sum_{\Delta_{n-1}} X_{\Delta_{n-1}} - J_2 \sum_{\Delta_n} W_{\Delta_n}
\end{equation}
with the gauge constraint (Gauss's law)
\begin{equation}
    G^0_{\Delta_{n-2}} = \prod_{\Delta_{n-1} \supset \Delta_{n-2}} X_{\Delta_{n-2}}=1.
\end{equation}
It is well-kwown that its Euclidean theory is $(n+1)$-dimensional Ising model (with a certain choice of $J_1$ and $J_2$) \cite{K79}:
\begin{equation}
    S_{\text{Ising}} (A_{n-1}) = -J \sum_{\Delta_{n} \subset Y} |\delta A_{n-1} (\Delta_{n})|
\end{equation}
where $A \in C^{n-1}(Y,\ZZ_2)$ is a $(n-1)$-cochain on the spacetime manifold $Y$, $|\delta A| = 0, 1$ gives $\delta A$ (mod $2$), and $J$ depends on $J_1$ and $J_2$. In this case, $S_{\text{Ising}}$ is invariant under the gauge transformation $A_{n-1} \ra A_{n-1} + \delta \Lambda_{n-2}$ for arbitrary $(n-2)$-cochain $\Lambda_{n-2} \in C^{n-2}(Y,\ZZ_2)$. Therefore, $S_{\text{Ising}}$ has no boundary anomaly under the standard Gauss's law.

Now, we propose a new class of $\ZZ_2$ lattice gauge theory:
\begin{equation}
    \begin{split}
        H &= - J_1 \sum_{\Delta_{n-1}} U_{\Delta_{n-1}} - J_2 \sum_{\Delta_{n}} W_{\Delta_{n}}
    \end{split}
\label{eq: modified Z2 lattice gauge theory}
\end{equation}
with the modified Gauss's law (gauge constraints) at $(n-2)$-simplices
\begin{equation}
    G_{\Delta_{n-2}} = \prod_{{\Delta_{n-1}} \supset {\Delta_{n-2}}} X_{\Delta_{n-1}}  (\prod_{{\Delta_{n-1}}^\prime} Z_{{\Delta_{n-1}}^\prime}^{\int \delta \boldsymbol {\Delta_{n-2}} \cup_{n-2} \boldsymbol {\Delta_{n-1}}^\prime})=1.
\label{eq: modified Gauss's law}
\end{equation}
This model describes a free fermion system, since it is dual to
\begin{equation}
    \begin{split}
        H_f &= - J_1 \sum_{\Delta_{n-1}} (-1)^{\int_E \Delta_{n-1}} i \g_{L(\Delta_{n-1})} \g'_{R(\Delta_{n-1})} - J_2 \sum_{\Delta_{n}} (- i \g_{{\Delta_{n}}} \g'_{\Delta_{n}}) \\
        &=  - J_1 \sum_{\Delta_{n-1}} S^E_{\Delta_{n-1}} - J_2 \sum_{\Delta_{n}} P_{\Delta_{n}}.
    \end{split}
\end{equation}
The modified Gauss's law \eqref{eq: modified Gauss's law} on a $(n-2)$-simplex $\Delta_{n-2}$, or equivalently on the dual $(n-2)$-cochain $\boldsymbol \Delta_{n-2}$, can be generalized to an arbitrary $(n-2)$-cochain $\lambda_{n-2} = \sum_{i} \boldsymbol \Delta^i_{n-2}$, the Gauss's law is
\begin{equation}
    \begin{split}
        1=&G_{\lambda_{n-2}} = \prod_{i} G_{\Delta^i_{n-2}} \\
        =& (\prod_{ \boldsymbol \Delta_{n-1} \in \delta \lambda_{n-2}} X_{\Delta_{n-1}})
        (\prod_{{\Delta_{n-1}}^\prime} Z_{{\Delta_{n-1}}^\prime}^{\int \delta \lambda_{n-2} \cup_{n-2} \boldsymbol {\Delta_{n-1}}^\prime}) \\
        & \quad \cdot (-1)^{\int \lambda_{n-2} \cup_{n-4} \lambda_{n-2} + \lambda_{n-2} \cup_{n-3} \delta \lambda_{n-2}}\\
    \end{split}
\label{eq: general gauge constraints}
\end{equation}
where the sign comes from anti-commutation of $X$ and $Z$ on the same simplex. This can be proved by induction:
\begin{enumerate}
    \item We first check for $\lambda_{n-2} = \boldsymbol \Delta_{n-2}$, where $\lambda_{n-2}$ contains a single $(n-2)$-simplex. We have $\boldsymbol \Delta_{n-2} \cup_{n-4} \boldsymbol \Delta_{n-2} + \boldsymbol \Delta_{n-2} \cup_{n-3} \delta \boldsymbol \Delta_{n-2}= 0$ by the definition of higher cup products since the vertices in \eqref{eq: definition of higher cup product} cannot match. For example, $\boldsymbol \Delta_{n-2}$ acts only nontrivial on a $(n-2)$-simplex with $(n-1)$ vertices, while $\boldsymbol \Delta_{n-2} \cup_{n-4} \boldsymbol \Delta_{n-2}$ has the input of $(n+1)$ vertices, which has 2 extra vertices at least. $\boldsymbol \Delta_{n-2}$ vanishes when it acts on any simplex with the extra vertices. The gauge constraint reduces the original form \eqref{eq: modified Gauss's law}.
    
    \item It is straightforward to check $G_{\lambda_{n-2}} G_{\lambda'_{n-2}}= G_{\lambda_{n-2}+\lambda'_{n-2}}$, using the recursive property of cup products:
\begin{equation}
    \alpha \cup_a \beta + \beta \cup_a \alpha = \alpha \cup_{a+1} \delta \beta + \delta \alpha \cup_{a+1} \beta + \delta (\alpha \cup_{a+1} \beta).
\end{equation}
\end{enumerate}

Consider now the following $(n-1)$-form gauge theory defined on a general triangulated $(n+1)$-dimensional manifold $Y$:
\begin{equation}
S(A_{n-1})= -\sum_{\Delta_{n} \subset Y} |\delta A_{n-1} (\Delta_{n})| + i \pi \int_Y (A_{n-1} \cup_{n-3} A_{n-1} + A_{n-1} \cup_{n-2} \delta A_{n-1}).
\label{eq: n-1 form gauge theory}
\end{equation}
where $A_{n-1}\in C^{n-1}(Y,\ZZ_2)$, and the gauge symmetry acts by $A_{n-1} \rightarrow A_{n-1} + \delta \Lambda_{n-2}$ for $\Lambda_{n-2} \in C^{n-2}(Y,\ZZ_2)$. The second term is the generalized "Steenrod square" defined in Ref.~\cite{LZW18}. The action is gauge-invariant up to a boundary term:
\begin{equation}
    \begin{split}
        &S(A_{n-1} + \delta \Lambda_{n-2}) - S(A_{n-1}) \\
        =& i \pi \int_{\partial Y} ( \Lambda_{n-2} \cup_{n-4} \Lambda_{n-2} + \Lambda_{n-2} \cup_{n-3} \delta \Lambda_{n-2} + \delta \Lambda_{n-2} \cup_{n-2} A_{n-1}) \\
        =& i \pi \int_{\partial Y} ( \Lambda \cup_{n-4} \Lambda + \Lambda \cup_{n-3} \delta \Lambda + \delta \Lambda \cup_{n-2} A)
    \end{split}
\label{eq:Ssimple boundary}
\end{equation}
where we have omited the subscript of $A_{n-1}$ and $\Lambda_{n-2}$ for simplicity.
This boundary term determines the Gauss's law for the wave-function $\Psi(A)$ on the spatial slice $M=\partial Y$:
\begin{equation}
\Psi(A + \delta \Lambda) = (-1)^{\omega(\Lambda,A)} \Psi (A)
\end{equation}
where $\omega(\Lambda,A) = \int_{M} ( \Lambda \cup_{n-4} \Lambda + \Lambda \cup_{n-3} \delta \Lambda + \delta \Lambda \cup_{n-2} A)$. The Gauss's law is the same as the gauge constraint \eqref{eq: general gauge constraints} if we identify $Z_{\Delta_{n-1}}$ as $(-1)^{A_{n-1} (\Delta_{n-1})}$ and $X_{\Delta_{n-1}}$ acts as the transformation $A_{n-1} \ra A_{n-1} + \boldsymbol \Delta_{n-1}$. The modified Gauss's law \eqref{eq: modified Gauss's law} represents the boundary anomaly of topological action \eqref{eq: n-1 form gauge theory} as we claimed.

In the following subsection, we derive the Euclidean action of the modified $\ZZ_2$ lattice gauge theory \eqref{eq: modified Z2 lattice gauge theory} explicitly, which is analogous to \eqref{eq: n-1 form gauge theory}. 

\subsection{Euclidean path integral of lattice gauge theories}

Start with the Hamiltonian of modified $\ZZ_2$ lattice gauge theory:
\begin{equation}
    \begin{split}
        H &= - J_1 \sum_{\Delta_{n-1}} U_{\Delta_{n-1}} - J_2 \sum_{\Delta_{n}} W_{\Delta_{n}} \\
        &= -J_1 \sum_{\Delta_{n-1}} X_{\Delta_{n-1}} (\prod_{{\Delta_{n-1}}^\prime} Z_{{\Delta_{n-1}}^\prime}^{\int \boldsymbol {\Delta_{n-1}}^\prime \cup_{n-2} \boldsymbol {\Delta_{n-1}}})
        - J_2 \sum_{\Delta_{n}} \prod_{\Delta_{n-1} \subset \Delta_n} Z_{\Delta_{n-1}}
    \end{split}
\end{equation}
with gauge constraints
\begin{equation}
    G_{\Delta_{n-2}} = \prod_{{\Delta_{n-1}} \supset {\Delta_{n-2}}} X_{\Delta_{n-1}}  (\prod_{{\Delta_{n-1}}^\prime} Z_{{\Delta_{n-1}}^\prime}^{\int \delta \boldsymbol {\Delta_{n-2}} \cup_{n-2} \boldsymbol {\Delta_{n-1}}^\prime})=1.
\end{equation}
The partition function is:
\begin{equation}
\begin{split}
\mathcal{Z} &=\text{Tr } e^{-\beta H} =  \text{Tr } T^M \\
\end{split}
\label{eq:partition function}
\end{equation}
where we use Trotter-Suzuki decomposition in imaginary time direction and $T$ is the transfer matrix defined as
\begin{equation}
T= \left(\prod_{\Delta_{n-2}} \delta_{G_{\Delta_{n-2}}, 1 }\right) e^{-\delta \tau H} .
\end{equation}
The first factor arises from the gauge constraints on the Hilbert space. The spacetime manifold consists of many time slices labelled by layers $\{i\}$. In the $i$th layer, we insert a complete basis (in Pauli matrix $Z_{\Delta_{n-1}}$):  $b^i_{n-1} \in C^{n-1}(M_n,\ZZ_2)$ (a $\ZZ_2$ field on each $\Delta_{n-1}$ of the spatial manifold $M_n$ such that $Z_{\Delta_{n-1}}=(-1)^{b^i_{n-1}(\Delta_{n-1})}$). The transfer matrix $T$ between the $i$th layer and the $(i+1)$th layer contains gauge constraints on every spatial $(n-2)$-simplex $\Delta_{n-2}$:
\begin{equation}
    \delta_{G_{\Delta_{n-2}}, 1 } = \frac{1+G_{\Delta_{n-2}}}{2}= \frac{1}{2} \sum_{a^{i+1/2}_{n-2}=0,1} (G_{\Delta_{n-2}})^{a^{i+1/2}_{n-2}}
\end{equation}
where we introduce the Lagrangian multiplier $a^{i+1/2}_{n-2} \in C^{n-2}(M_n,\ZZ_2)$ (a $\ZZ_2$ field living on each $\Delta_{n-2}$ of the spatial manifold $M_n$). Notice that $a^{i+1/2}_{n-2}$ defined between two time slices lives on the spatial $(n-2)$-simplex $\Delta_{n-2}$, which can be interpreted as the spacetime $(n-1)$-simplex between the two layers. From the similar calculation in Ref.~\cite{CK18}, we have
\begin{equation}
    \mathcal{Z} = \sum_{ \{ \{ a^{i+1/2}_{n-2}\} , \{ b^i_{n-1}\} \} } \exp ([S_{\text{Ising}} + S_{\text{top}}] (\{ \{ a^{i+1/2}_{n-2}\} , \{ b^i_{n-1}\} \}) )
\end{equation}
where
\begin{equation}
    \begin{split}
        &S_{\text{Ising}}(\{ \{ a^{i+1/2}_{n-2}\} , \{ b^i_{n-1}\} \})  \\
        =& \sum_i \left(-J_s \sum_{\Delta_n} |\delta b^i_{n-1}(\Delta_n)| - J_\tau \sum_{\Delta_{n-1}} | \left[b^i_{n-1} + b^{i+1}_{n-1} + \delta a^{i+1/2}_{n-2}\right] (\Delta_{n-1})| \right)
    \end{split}
\label{eq: S Ising}
\end{equation}
and
\begin{equation}
    \begin{split}
        &S_{\text{top}}(\{ \{ a^{i+1/2}_{n-2}\} , \{ b^i_{n-1}\} \})  \\
        = i \pi \sum_i &\int_{M_n}  a^{i+1/2}_{n-2} \cup_{n-4} a^{i+1/2}_{n-2} + a^{i+1/2}_{n-2} \cup_{n-3} \delta a^{i+1/2}_{n-2} \\
        &+ \delta a^{i+1/2}_{n-2} \cup_{n-2} b^{i+1}_{n-1} + b^{i}_{n-1} \cup_{n-2} ( b^{i}_{n-1} + b^{i+1}_{n-1} + \delta a^{i+1/2}_{n-2}).
    \end{split}
\label{eq: S top}
\end{equation}
Here $J_s, J_\tau$ are constants depending on $J_1,J_2,\delta \tau$ in the original Hamiltonian and we assume $J_s = J_\tau = J$ for simplicity. $|\cdots|$ gives the argument's parity $0$ or $1$. The gauge transformations act as
\begin{equation}
    \begin{split}
        b^i_{n-1} &\rightarrow b^i_{n-1} + \delta \lambda^i,\\
        a^{i+1/2}_{n-2} &\rightarrow a^{i+1/2}_{n-2} + \delta \mu^i + \lambda^i + \lambda^{i+1},
    \end{split}
\label{eq:gauge transformation}
\end{equation}
where $\lambda^i$ are arbitrary $(n-2)$-cochains and $\mu^i$ are arbitrary $(n-3)$-cochains.

If we interpret $a^{i+1/2}_{n-2}$ as spacetime $(n-1)$-cochains, we can rewrite
\begin{equation}
    \{ \{ a^{i+1/2}_{n-2}\} , \{ b^i_{n-1}\} \}  \ra A_{n-1} \in C^{n-1}(Y, \ZZ_2),
\end{equation}
 which is a $\ZZ_2$ field living on $(n-1)$-simplices in spacetime manifold $Y$. It is natural to write $S_{\text{Ising}}$ in \eqref{eq: S Ising} as
\begin{equation}
     S_{\text{Ising}} = -\sum_{\Delta_{n} \subset Y} |\delta A_{n-1} (\Delta_{n})|.
\end{equation}
The spacetime manifold $Y = M_n \times [-\infty,0]$ (spatial and temporal parts) is not a triangulation, since we only triangularize the spatial manifold $M_n$ under the discretized time. The (higher) cup products are not well-defined in $Y$. However, we can still write an expression  
\begin{equation}
     S_{\text{top}} = i \pi \int_{Y'} (A_{n-1} \cup_{n-3} A_{n-1} + A_{n-1} \cup_{n-2} \delta A_{n-1}).
\label{eq: S top spacetime}
\end{equation}
in $(n+1)$-dimensional triangulation $Y'$ such that $Y'$ is a refinement of $Y$. We can check that \eqref{eq: S top} and \eqref{eq: S top spacetime} produce the same boundary term under gauge transformations.

\section{Conclusions}

We have extended the the exact bosonization \eqref{eq: 2d boson-fermion duality} in 2d and \eqref{eq: 3d boson-fermion duality} in 3d to arbitrary dimensions. The dictionary for $n$-dimensional boson-fermion duality is given in \eqref{eq: nd boson-fermion duality}. This bosonization is a duality between any fermionic system in arbitrary $n$ spatial dimensions and $(n-1)$-form $\mathbb{Z}_2$ gauge theories in $n$ dimensions with gauge constraints (the modified Gauss's law). This map preserves locality: every local even fermionic observable is mapped to a local gauge-invariant bosonic operator. The formula has an explicit dependence on the second Stiefel-Whitney class of the manifold, and a choice of spin structure is needed. As a side product, we discover a new formula \eqref{eq: S P sign} for Stiefel-Whitney homology classes on lattices. In the Euclidean picture, we have shown that the Euclidean path integral of the $n$-dimensional $\ZZ_2$ gauge theory with modified Gauss's law is the $(n+1)$-dimensional Ising model with an additional topological Steenrod square \eqref{eq: n-1 form gauge theory} term.

\section*{Acknowledgement}
Y.C. thanks Po-Shen Hsin and his advisor Prof. Anton Kapustin for many very helpful discussions. Y.C. also thanks Tyler Ellison and Nathanan Tantivasadakarn for their useful feedback. This research was supported in part by the U.S. Department of Energy, Office of Science, Office of High Energy Physics, under Award Number de-sc0011632. Anton Kapustin was also supported by the Simons Investigator Award.

\appendix

\section{A formula for Stiefel-Whitney homology classes} \label{sec: formula for Stiefel-Whitney homology classes}

In this section, we prove Lemma \ref{lemma: w2}, Eq.~\eqref{eq: c Delta}.
First, let us recall the theorem proved in \cite{GT76}.
Let $s$ be a $p$-simplex, say $s = \langle v_0, v_1,\dots, v_p \rangle$. Let $k$ be another simplex which has s as a face; i.e., $s \subset k$ ($s$ may be equal to $k$). Let
\begin{equation}
    \begin{split}
        B_{-1} &= \text{set of vertices of $k$ less than $v_0$,}\\
        B_{0} &= \text{set of vertices of $k$ between $v_0$ and $v_1$,} \\
        B_{m} &= \text{set of vertices of $k$ between $v_m$ and $v_{m+1}$,} \\
        B_{p} &= \text{set of vertices of $k$ greater than $v_p$.} 
    \end{split}
\end{equation}
We say that $s$ is regular in $k$, if $ \# (B_m) = 0$ for every odd $m$. Let $\partial_p (k)$ denote the mod 2 chain which consists of all $p$-dimensional simplices $s$ in $k$ so that $s$ is regular in $k$. For example, $\langle 012 \rangle$ and $\langle 023 \rangle$ are regular in $\langle 0123 \rangle$ and therefore $\partial_2 (\langle 0123 \rangle) = \langle 012 \rangle + \langle 023 \rangle$. The theorem is \cite{GT76}:
\begin{theorem}
    $\sum_{k| \dim k \geq {(n-2)} } \partial_{n-2}(k)$ is a $(n-2)$-chain which represents $w_2$.
\label{eq: theorem}
\end{theorem}
In particular, for any $n'$-simplex $\Delta_{n'} = \langle 0 \dots n' \rangle$, all $(n'-1)$-simplices regular in $\Delta_{n'}$ are
\begin{equation}
    \langle 0 \dots \hat{i} \dots n \rangle ~\forall i \in \text{odd}
\end{equation}
and all $(n'-2)$-simplices regular in $\Delta_{n'}$ are
\begin{equation}
    \langle 0 \dots \hat{i} \dots \hat{j} \dots n \rangle ~\forall i \in \text{odd}, j \in \text{even}, i<j.
\end{equation}
We now use this theorem to prove lemma \ref{lemma: w2}.

\begin{proof}[Proof of Lemma \ref{lemma: w2}]
For every $(n-2)$-simplex $\Delta_{n-2}$, it is regular in itself. This contributes the $1$ in the coefficient of $c(\Delta_{n-2})$ in \eqref{eq: c Delta}.

For every $(n-1)$-simplex $\Delta_{n-1}$, it is a boundary of two $n$-simplices $\Delta^L_n$ and $\Delta^R_n$, with $\Delta_{n-1}$ being an outward boundary of $\Delta^L_n$ and an inward boundary of $\Delta^R_n$. We define that $\Delta_{n-1}$ belongs to $\Delta^R_n$ and the summation of $\dim k= n-1,~n$ in theorem \ref{eq: theorem} can be written as:
\begin{equation}
    \begin{split}
        &\sum_{\Delta_{n-1}} \partial_{n-2} (\Delta_{n-1}) + \sum_{\Delta_{n}} \partial_{n-2} (\Delta_{n}) \\
        &= \sum_{\Delta_{n}} \left[ \partial_{n-2} (\Delta_{n}) + \sum_{\Delta_{n-1} \in \Delta_{n} | \Delta_{n-1} \text{ is inward} } \partial_{n-2} (\Delta_{n-1})\right].
    \end{split}
\end{equation}
If $\Delta_n = \langle 0 \dots n \rangle$ is $``+"$-oriented, the terms in the summation is
\begin{equation}
    \begin{split}
        &\partial_{n-2}(\langle 0 \dots n \rangle) + \sum_{0 \leq i \leq n | i \in  \text{odd}} \partial_{n-2}(\langle 0 \dots \hat{i} \dots n \rangle) \\
        =& \sum_{ i,j| i<j,~i \in  \text{odd},~j \in \text{even}} \langle 0 \dots \hat{i} \dots \hat{j} \dots n \rangle \\
        &+ \sum_{0 \leq i \leq n | i \in  \text{odd}} (\sum_{j<i | j\in \text{odd}} \langle 0 \dots \hat{j} \dots \hat{i} \dots n \rangle + \sum_{j>i | j\in \text{even}} \langle 0 \dots \hat{i} \dots \hat{j} \dots n \rangle) \\
        =& \sum_{ i,j| i<j,~i \in  \text{odd},~j \in \text{odd}} \langle 0 \dots \hat{i} \dots \hat{j} \dots n \rangle
        \end{split}
\label{eq: + contribution}
\end{equation}
where we have used the definition of regular simplex defined above.
Similarly, we can derive that if $\Delta_n = \langle 0 \dots n \rangle$ is $``-"$-oriented, the term is
\begin{equation}
    \sum_{ i,j| i<j,~i \in  \text{even},~j \in \text{even}} \langle 0 \dots \hat{i} \dots \hat{j} \dots n \rangle.
\label{eq: - contribution}
\end{equation}
Combining \eqref{eq: + contribution} and \eqref{eq: - contribution} with the $1$ from $\dim k = n-2$ in theorem \ref{eq: theorem}, we have
\begin{equation}
        w_2 = \sum_{\Delta_{n-2}} c(\Delta_{n-2}) \Delta_{n-2}
\end{equation}
where
\begin{equation}
    \begin{split}
        &c(\Delta_{n-2})= \\
        & 1 + \sum_{``-"\text{-oriented } \Delta_n = \langle 0 \dots n \rangle } \sum_{j_1 < j_2 | j_1, j_2 \in \text{even}} \boldsymbol \Delta_{n-2} ( \langle 0 \cdots \hat{j_1} \cdots \hat{j_2} \cdots n \rangle) \\
        &+ \sum_{``+"\text{-oriented } \Delta_n = \langle 0 \dots n \rangle } \sum_{k_1 < k_2 | k_1, k_2 \in \text{odd}} \boldsymbol \Delta_{n-2} ( \langle 0 \cdots \hat{k_1} \cdots \hat{k_2} \cdots n \rangle).
    \end{split}
\end{equation}

\end{proof}

\bibliographystyle{unsrt}
\bibliography{bibliography.bib}

\begin{thebibliography}{10}

\bibitem{CKR18}
Yu-An Chen, Anton Kapustin, and Djordje Radicevic.
\newblock Exact bosonization in two spatial dimensions and a new class of
  lattice gauge theories.
\newblock {\em Annals of Physics}, 393:234 -- 253, 2018.

\bibitem{CK18}
Yu-An {Chen} and Anton {Kapustin}.
\newblock {Bosonization in three spatial dimensions and a 2-form gauge theory}.
\newblock {\em arXiv e-prints}, page arXiv:1807.07081, Jul 2018.

\bibitem{GK16}
Davide Gaiotto and Anton Kapustin.
\newblock Spin tqfts and fermionic phases of matter.
\newblock {\em International Journal of Modern Physics A}, 31(28n29):1645044,
  2016.

\bibitem{S47}
N.~E. Steenrod.
\newblock Products of cocycles and extensions of mappings.
\newblock {\em Annals of Mathematics}, 48(2):290--320, 1947.

\bibitem{LZW18}
Tian {Lan}, Chenchang {Zhu}, and Xiao-Gang {Wen}.
\newblock {Fermion decoration construction of symmetry protected trivial orders
  for fermion systems with any symmetries $G_f$ and in any dimensions}.
\newblock {\em arXiv e-prints}, page arXiv:1809.01112, Sep 2018.

\bibitem{KT17}
Anton Kapustin and Ryan Thorngren.
\newblock Fermionic spt phases in higher dimensions and bosonization.
\newblock {\em Journal of High Energy Physics}, 2017(10):80, 2017.

\bibitem{W17}
Xiao-Gang Wen.
\newblock Exactly soluble local bosonic cocycle models, statistical
  transmutation, and simplest time-reversal symmetric topological orders in 3+1
  dimensions.
\newblock {\em Phys. Rev. B}, 95:205142, May 2017.

\bibitem{CET20}
Yu-An Chen, Tyler~D Ellison, and Nathanan Tantivasadakarn.
\newblock Disentangling supercohomology symmetry-protected topological phases
  in three spatial dimensions.
\newblock {\em arXiv preprint arXiv:2008.05652}, 2020.

\bibitem{BK02}
Sergey~B. Bravyi and Alexei~Yu. Kitaev.
\newblock Fermionic quantum computation.
\newblock {\em Annals of Physics}, 298(1):210 -- 226, 2002.

\bibitem{B05}
R.~C. Ball.
\newblock Fermions without fermion fields.
\newblock {\em Phys. Rev. Lett.}, 95:176407, Oct 2005.

\bibitem{VC05}
F~Verstraete and J~I Cirac.
\newblock Mapping local hamiltonians of fermions to local hamiltonians of
  spins.
\newblock {\em Journal of Statistical Mechanics: Theory and Experiment},
  2005(09):P09012--P09012, sep 2005.

\bibitem{B19}
Tom Banks.
\newblock Fermi/pauli duality in arbitrary dimension.
\newblock {\em arXiv preprint arXiv:1908.10453}, 2019.

\bibitem{T20}
Sri Tata.
\newblock Geometrically interpreting higher cup products, and application to
  combinatorial pin structures.
\newblock {\em arXiv preprint arXiv:2008.10170}, 2020.

\bibitem{K79}
John~B. Kogut.
\newblock An introduction to lattice gauge theory and spin systems.
\newblock {\em Rev. Mod. Phys.}, 51:659--713, Oct 1979.

\bibitem{GT76}
Richard~Z. Goldstein and Edward~C. Turner.
\newblock A formula for stiefel-whitney homology classes.
\newblock {\em Proceedings of the American Mathematical Society},
  58(1):339--339, jan 1976.

\end{thebibliography}

\end{document}